\documentclass{article}
\usepackage[letterpaper, top=1in, bottom=1in, left=0.8in, right=0.8in]{geometry}
\usepackage{graphicx}
\usepackage[hyphens]{url}
\usepackage{amsfonts,amsmath,amstext,amssymb,amsthm}
\usepackage{enumerate}
\usepackage{tabularx}
\usepackage{mcode}
 \usepackage{float}
 \usepackage{subcaption}
 \usepackage{authblk}
 \usepackage{algorithm, algpseudocode}
 \usepackage[normalem]{ulem}
 \usepackage{subcaption}

\newtheorem{theorem}{Theorem}
\newtheorem{proposition}{Proposition}
\newtheorem{remark}{Remark}

\newtheorem{opt}{Optimization problem}
\newtheorem{highopt}{High-level optimization problem}

\def\BibTeX{{\rm B\kern-.05em{\sc i\kern-.025em b}\kern-.08em
    T\kern-.1667em\lower.7ex\hbox{E}\kern-.125emX}}

\newcommand{\bmat}[1]{\begin{bmatrix}#1\end{bmatrix}}
\newcommand\norm[1]{\left\lVert#1\right\rVert}
\newcolumntype{M}[1]{>{\centering\arraybackslash}m{#1}}

\newcommand{\He}[1]{{\color{black}  #1}}

\newcommand{\todo}[1]{{\color{red}  #1}}

\date{\vspace{-5ex}}

\begin{document}
\title{\LARGE \bf Backward Reachability for Polynomial Systems on A Finite Horizon}
\author{
	He Yin \thanks{He Yin is a Graduate Student in the Department of 
		Mechanical Engineering at the University of California, Berkeley 
		{\tt\small he\_yin@berkeley.edu}} , 
	Murat Arcak \thanks{Murat Arcak is a Professor in the Department of Electrical Engineering
		and Computer Sciences at the University of California, Berkeley
		{\tt\small arcak@berkeley.edu}} ,
	  Andrew Packard \thanks{Andrew Packard is a Professor in the Department of
		Mechanical Engineering at the University of California, Berkeley
		{\tt\small apackard@berkeley.edu}} ,
	  Peter Seiler \thanks{Peter Seiler is an Associate Professor in the Department of Aerospace Engineering and Mechanics at the University of
	 	Minnesota
	 	{\tt\small seile017@umn.edu}}%
}

\maketitle

\begin{abstract}
A method is presented to obtain an inner-approximation of the backward reachable set (BRS) \He{of a given target tube}, along with an \He{admissible} controller that \He{maintains trajectories inside this tube}. The proposed optimization algorithms are formulated as nonlinear optimization problems, which are decoupled into tractable subproblems and solved by an iterative algorithm using the \He{polynomial} S-procedure and sum-of-squares techniques. This framework is also extended to uncertain nonlinear systems with $\mathcal{L}_2$ disturbances and $\mathcal{L}_{\infty}$ parametric uncertainties. The effectiveness of the method is demonstrated on several nonlinear robotics and aircraft systems with control saturation.
\end{abstract}

\section{Introduction}
\label{sec:introduction}

The backward reachable set (BRS) is the set of all initial conditions \He{whose successors} can be \He{maintained} safely inside a \He{given time-varying state constraint set (``target tube'')} using an admissible controller \He{while satisfying control constraints}. \He{The BRS and the accompanying controller are of great importance for safety-critical systems.} \He{In this paper, we address the computation of an inner-approximation to the BRS and construction of an explicit feedback control action (as a state-feedback) on a finite-time horizon. We focus on problems with \emph{finite-time horizons}, since in many practical settings, systems only undergo finite-time trajectories, such as robotic systems and space launch / re-entry vehicles. }

 Lyapunov-based methods for the finite-horizon BRS computation are pursued in \cite{Majumdar:13}, where reference tracking controllers are designed to maximize the size of the BRS for error states, and in \cite{Majumdar:17}, where the goal is to compute a reference tracking controller by minimizing the size of an invariant funnel of the tracking error. \He{The computational approach put forth in \cite{Majumdar:13} and \cite{Majumdar:17} involves gridding in time, with S-procedure and sum-of-squares (SOS) techniques handling the state-space containments.  In \cite{Ian:05}, gridding is used in both space and time.}
 
 \He{A related computation that does not rely on gridding is considered in \cite{Henrion:14} and \cite{Majumdar:14}, where the BRS is outer-approximated by taking the complement of the initial set from which no trajectory	is able to reach the target set for any admissible inputs. This yields an infinite-dimensional linear program, and a sequence of finite-dimensional convex problems, along with results that prove convergence (from outside) to the true BRS, as more computational resources are employed.	Reference \cite{Henrion:14} proves that no suitable control action exists for initial	conditions outside the BRS outer-approximation. In contrast, \cite{Majumdar:14} modifies the formulation	and produces explicit control laws which will be suitable for some of the	points within the BRS outer-approximation. In addition, the obtained control laws will only approximately satisfy any given control constraints.}
 
 \He{The main contributions of the current paper are: (1) to explicitly synthesize a control law and an associated BRS \emph{inner-approximation},	(2) to accommodate various	sources of uncertainty simultaneously, including $\mathcal{L}_2$ disturbances  and $\mathcal{L}_{\infty}$ parametric uncertainties, (3) to present an iterative algorithm based on SDPs, with the guarantee that the certified inner-approximation to the BRS grows with each iteration.} \He{The results in this paper are complementary to those in \cite{Henrion:14}, \cite{Majumdar:14}, because we provide inner-approximations in which every point is guaranteed to lie in the BRS, as well as an explicit controller.} \He{By also avoiding} gridding of the time, state space or control space, \He{we provide a formal guarantee that the trajectories starting inside the inner-approximation remain inside the target tube.}
 
 \He{To enable these contributions,} the paper introduces a class of dissipation inequalities with associated “reachability storage functions”, whose sub-level sets characterize the inner-approximations to the BRS.   The polynomial S-procedure \cite{Parrilo:00} and SOS for polynomial non-negativity are used, \He{expressing the problem as a nonconvex optimization. The decision variables consist of a reachability storage function, a polynomial control law, and various S-procedure polynomial certificates. A tractable algorithm results, with further conservatism, by decoupling the original formulation into an iterative, two-way search between reachability storage functions and control laws, which are convex and quasiconvex problems, respectively.  \He{The use of dissipation inequalities also allows us to accommodate} various forms of disturbances and model uncertainty.}
 
 \He{Dissipation inequalities have also been applied to the}  related problem of region of attraction (ROA) estimation \He{which, however, is an infinite-time horizon problem. Associated with an equilibrium point, the ROA is the largest invariant set such that all trajectories starting inside converge to the equilibrium as $t \rightarrow \infty$. The literature on ROA estimation includes methods to search for a Lyapunov certificate for both stability and invariance \cite{Topcu:08} \cite{Weehong:08} \cite{Chesi:04} \cite{Fan:06} \cite{Andrea:19} and to synthesize a control law to expand the inner-approximation of the ROA \cite{Jarvis:05}.} 
 
 The conference version \cite{Yin:18} of this paper decomposes the control synthesis process into two steps: constructing storage functions first, and then computing control laws using the obtained storage functions through quadratic programs. The current paper presents a single-step design and accommodates control saturation, which is not addressed in \cite{Yin:18}. In addition, \cite{Yin:18} considers only a terminal target set, whereas this paper addresses a target tube. 
 In a separate publication \cite{He:18}, we have studied \emph{forward} reachable sets without control design.

\section{NOTATION}
$\mathbb{R}^{m\times n}$ and \He{$\mathbb{S}^{n\times n}$} denote the set of $m$-by-$n$ real matrices and $n$-by-$n$ real, symmetric matrices. $\mathbb{R}^m$ is the set of $m \times 1$ vectors whose elements are in $\mathbb{R}$. $\mathcal{C}^1$ is the set of differentiable functions with continuous derivative. $\mathcal{L}_2^m$ is the space of $\mathbb{R}^m$-valued measureable functions $f: [0, \infty) \rightarrow \mathbb{R}^m$, with $\norm{f}^2_2 := \int_0^{\infty} f(t)^Tf(t) dt < \infty$. Define $\norm{r}^2_{2,T} := \int_0^T r^T(t)r(t) dt.$ Associated with $\mathcal{L}_2^m$ is the extended space $\mathcal{L}_{2e}^m$, consisting of functions whose truncation $f_T(t) := f(t)$ for $t \le T$; $f_T(t) := 0$ for $t > T$, is in $\mathcal{L}_2^m$ for all $T > 0.$ For $\xi \in \mathbb{R}^n$, $\mathbb{R}[\xi]$ represents the set of polynomials in $\xi$ with real coefficients, and $\mathbb{R}^m[\xi]$ and $\mathbb{R}^{m \times p}[\xi]$ to denote all vector and matrix valued polynomial functions. The subset $\He{\Sigma[\xi] := \left\{\pi = \sum_{i=1}^M \pi_i^2 : M\ge 1, \pi_i \in \mathbb{R}[\xi]\right\}}$ of $\mathbb{R}[\xi]$ is the set of sum-of-squares (SOS) polynomials. For $\eta \in \mathbb{R}$, and continuous $r: \mathbb{R}^n \rightarrow \mathbb{R}$, $\Omega_{\eta}^r := \{x \in \mathbb{R}^n : r(x) \le \eta\}.$ For $\eta \in \mathbb{R}$, and continuous $r : \mathbb{R} \times \mathbb{R}^{n} \rightarrow \mathbb{R}$, define $\Omega_{t,\eta}^{r} := \{x \in \mathbb{R}^n : r(t,x) \le \eta\}$, a $t$-dependent set.

\He{In several places, a relationship between an algebraic condition on some real variables and input/output/state properties of a dynamical system is claimed. We use the same symbol for a particular real variable in the algebraic statement as well as the corresponding signal in the dynamical system. }

\section{Reachability Storage Functions and Control Synthesis}
Consider a time-varying, nonlinear system with affine dependence on the control input $u$:
\begin{align}
\dot{x}(t) = f(t,x(t)) + g(t,x(t))u(t), \label{eq:system1}
\end{align} 
with $x(t) \in \mathbb{R}^n$, $u(t) \in \mathbb{R}^m$, and mappings $f: \mathbb{R} \times \mathbb{R}^n \rightarrow \mathbb{R}^n$, $g: \mathbb{R} \times \mathbb{R}^n \rightarrow \mathbb{R}^{n \times m}$ continuous in $t$ and locally Lipschitz in $x$.

Denote $\phi(t; t_0, x_0, u)$ as the solution to the system (\ref{eq:system1}) at time $t \ (t_0 \leq t \leq T)$, from the initial condition $x_0$, under the control action $u(t)$. The function $r(t,x)$ is specified by the analyst, defining a target tube, $\Omega_{t,0}^r$. The target tube embodies time-varying state constraints, which are used to exclude unsafe regions, shape the trajectories $\phi$ and specify the desired set of states. \He{The BRS is defined as a set of states: $\{\xi \in \mathbb{R}^n : \exists u(\cdot), \ s.t. \ \phi(t; t_0, \xi, u) \in \Omega_{t,0}^r , \forall t \in [t_0, T]\}$.} 

In this paper, we consider an explicit time-varying, state-feedback control.
Let $k: \mathbb{R} \times \mathbb{R}^n \rightarrow \mathbb{R}^m$ define a memoryless, time-varying state feedback control by $u(t) = k(t,x(t))$. 

\He{An inner-approximation to the BRS is characterized by the level sets of ``reachability storage functions'' $V$ satisfying the conditions in the following proposition.}

\begin{proposition}\label{prop1}
	Given system (\ref{eq:system1}), initial time $t_0$, terminal time $T \ge t_0$, \He{a function $r$ and associated target tube $\Omega_{t,0}^{r}$}, and $\gamma \in \mathbb{R}$, if there exists a $\mathcal{C}^1$ function $V: \mathbb{R} \times \mathbb{R}^n \rightarrow \mathbb{R}$  and a control law $k: \mathbb{R} \times \mathbb{R}^n \rightarrow \mathbb{R}^m$ \He{that is continuous in $t$ and locally Lipschitz in $x$}, such that
	\begingroup
	\allowdisplaybreaks
	\begin{align}
	&\frac{\partial V}{\partial t} + \frac{\partial V}{\partial x}\left( f(t,x) + g(t,x)k(t,x) \right) \leq 0,  \forall (t, x) \in [t_0, T] \times \mathbb{R}^n, \text{and} \tag{A.1} \label{eq:A1} \\
	&\Omega^V_{t,\gamma} \subseteq \Omega_{t,0}^{r}, \ \text{for all} \ t \in [t_0, T], \tag{A.2} \label{eq:A2}
	\end{align}
	\endgroup
	then under the control law $k$, any trajectory of (\ref{eq:system1}) with initial condition $x(t_0) \in \Omega_{t_0, \gamma}^V$, satisfies $\phi(t; t_0, x(t_0), k) \in \Omega_{t, 0}^{r}$,  for all $t \in [t_0, T]$, i.e. all the trajectories remain inside the target tube. Such a function $V$ is called a {\it reachability storage function}.
\end{proposition}

The set $\Omega_{t_0, \gamma}^V$ is an inner-approximation of the BRS for the given target tube and the initial time, associated with the control law $k$. For simplicity, we will use $x(t)$ to represent the state trajectories $\phi(t; t_0, x_0, k)$ in the rest of the paper. Proposition \ref{prop1} follows from a simple dissipation argument. Integrating constraint (\ref{eq:A1}) from $t_0$ to $t$ yields $V(t,x(t)) \leq V(t_0,x(t_0))$. Thus it follows from $x(t_0) \in \Omega_{t_0,\gamma}^V$ that $V(t,x(t)) \le \gamma$. Assumption (A.2) then implies that $x(t)$ stays in the target tube for all $t \in [t_0, T]$.

In some cases, the target tube might be defined only at the terminal time, i.e., the only constraint is $x(T) \in \Omega_{T,0}^r$, for all $x(t_0) \in \Omega_{t_0,\gamma}^V$. The set $\Omega_{T,0}^r$ is called the terminal target set, and it can be addressed by enforcing (\ref{eq:A2}) to hold only for $t = T$, which is equivalent to
\begin{align}
\Omega_{T,\gamma}^V \subseteq \Omega_{T,0}^r. \tag{A.3} \label{eq:A3}
\end{align}
Here $\Omega_{t_0,\gamma}^V$ is the inner-approximated BRS from the terminal target set $\Omega_{T,0}^r$. For simplicity, define $r_T(x) := r(T,x)$, and rewrite the terminal target set as $\Omega_{0}^{r_T}$.

\subsection{Local Synthesis}
\He{Constraint (\ref{eq:A1}) is conservative in that it holds throughout the state space, but the conclusion of Proposition \ref{prop1} only applies to a subset, namely $\Omega_{t,\gamma}^V$. By restricting where (\ref{eq:A1}) must hold, we obtain a less conservative local condition.}
 
\begin{theorem} \label{thm1}
	Given system (\ref{eq:system1}), initial time $t_0$, terminal time $T \ge t_0$, \He{a function $r$ and associated target tube $\Omega_{t,0}^{r}$}, and $\gamma \in \mathbb{R}$, if there exists a $\mathcal{C}^1$ function $V: \mathbb{R} \times \mathbb{R}^n \rightarrow \mathbb{R}$, and a control law $k: \mathbb{R} \times \mathbb{R}^n \rightarrow \mathbb{R}^m$ \He{that is continuous in $t$ and locally Lipschitz in $x$}, such that for all $t \in [t_0, T]$, the following two constraints hold, 
	\begin{align}
	&\Omega_{t,\gamma}^V \subseteq \left\{x\in \mathbb{R}^n \middle \vert\frac{\partial V}{\partial t} + \frac{\partial V}{\partial x} (f(t,x) + g(t,x)k(t,x)) \leq 0\right\}, \tag{B.1} \label{eq:B1} \\
	&\Omega^V_{t,\gamma} \subseteq \Omega_{t,0}^{r}, \tag{B.2} \label{eq:B2}
	\end{align}
	then under the control law $k$, any trajectory with initial condition $x(t_0) \in \Omega_{t_0, \gamma}^V$, satisfies $x(t) \in \Omega_{t,0}^{r}$ for all $t \in [t_0, T]$. 
\end{theorem}

Again, $\Omega_{t_0, \gamma}^V$ is an inner-approximation of the BRS for the given target tube and the initial time, associated with the control law $k$. This theorem is a special case of Theorem \ref{thm3} stated later, and hence the proof of Theorem \ref{thm1} is omitted.

\He{
\begin{remark}
	If the same constant is added to $V$ and $\gamma$, the conditions (\ref{eq:B1}) and (\ref{eq:B2}) are unchanged. Hence $\gamma$ can be fixed to any specific value. However, $\gamma$ is retained here as it is exploited by Algorithm \ref{alg:alg1} introduced later in the paper.
\end{remark}
}

Since a less conservative inner-approximation is preferable, the volume of $\Omega_{t_0, \gamma}^V$ becomes the objective (to be maximized), resulting in an optimization problem, \He{where $\gamma$ is either fixed or a decision variable.}
\begin{highopt}\label{highopt1}($hi$-$opt_1$)
    \begingroup
	\allowdisplaybreaks
	\begin{align}
	&\sup_{V,k} \ \text{volume}(\Omega_{t_0,\gamma}^V) \nonumber \\
    & s.t. \ (\ref{eq:B1}) \ \text{and} \ (\ref{eq:B2}) \ \text{hold for all} \ t \in [t_0, T]  \nonumber 
	\end{align}
	\endgroup
\end{highopt}

\subsection{Modifications for Control Saturation}
In practice, the magnitude of control inputs to any system cannot be arbitrarily large, so we introduce constraints on the magnitude of control $u$.  Specifically, assume the set of control constraints is given as a \He{time- and state-varying polytope: 
\begin{align}
\mathcal{U}(t,x) := \{u \in \mathbb{R}^m: A(t,x)u \leq b(t,x)\}, \nonumber
\end{align}
where $A(t,x) \in \mathbb{R}^{n_p \times m}[t,x]$ and $b(t,x) \in  \mathbb{R}^{n_p}[t,x]$ are given matrix and vector valued polynomial functions, $n_p$ is the number of constraints on $u$,} and the symbol ``$\leq$" represents componentwise inequality. To take control saturation into account as in \cite{Jarvis:05}, we impose additional constraints for $V$ and $k$: for all $t \in [t_0, T]$, \He{
\begin{align}
&\Omega_{t,\gamma}^V \subseteq \{x \in \mathbb{R}^n : A(t,x)k(t,x) \leq b(t,x)\} \tag{C.1} \label{eq:C1}.
\end{align}
}This ensures while $x(t)$ lies in the funnel $\Omega_{t,\gamma}^V$, the control input $u$ derived from the control law $k$ remains within $\mathcal{U}(t,x)$.
 
Combining the high-level optimization problem $hi$-$opt_{\ref{highopt1}}$ and constraints (\ref{eq:C1}) yields a synthesis optimization that accounts for actuator limits.
\begin{highopt} \label{highopt2} ($hi$-$opt_2$)
	\begingroup
	\allowdisplaybreaks
	\begin{align}
	&\sup_{V,k} \ \text{volume}(\Omega_{t_0,\gamma}^V) \nonumber \\
	& s.t.\ \Omega_{t,\gamma}^V \subseteq \bigg\{x\in \mathbb{R}^n \bigg \vert\frac{\partial V}{\partial t} + \frac{\partial V}{\partial x} (f(t,x) + g(t,x)k(t,x)) \leq 0\bigg\}, \ \forall \ t \in [t_0, T], \tag{D.1} \label{eq:D1} \\
	& \Omega_{t,\gamma}^V \subseteq \Omega_{t,0}^{r}, \ \forall \ t \in [t_0, T], \tag{D.2} \label{eq:D2} \\
	&\Omega_{t,\gamma}^V \subseteq \{x \in \mathbb{R}^n : A(t,x)k(t,x) \leq b(t,x)\}, \forall \ t \in [t_0, T]. \tag{D.3} \label{eq:D3} 
	\end{align}
	\endgroup
\end{highopt}

\subsection{Reformulating as a Polynomial Optimization}
 As written, $hi$-$opt_{\ref{highopt1}}$ and $hi$-$opt_{\ref{highopt2}}$ involve many set containment constraints with a storage function and a control law as decision variables. The most common way of certifying set containments is the S-procedure, along with a method to check non-negativity. To check non-negativity, SOS relaxation is widely used when the functions are restricted to polynomials. Therefore, for practical computation, we restrict the system model, control law and storage function to be polynomials, i.e., $f \in \mathbb{R}^n[t,x], g \in \mathbb{R}^{n \times m}[t,x], k \in \mathbb{R}^m[t,x]$ and $V \in \mathbb{R}[t,x]$. \He{Note that it is sometimes possible to represent nonlinear system equations
with polynomials using combinations of change-of-variables, Taylor’s theorem and least squares regression.} \He{The error on the polynomial approximation can be handled by Theorem \ref{thm3} and is illustrated in the example \ref{ex_game}.} Since the formulation involves finite horizon problems on $[t_0, T]$, the function $h(t) := (t - t_0)(T - t)$ is important in the \He{S-procedure} as it is nonnegative on this interval. With these ideas, we reformulate constraints (\ref{eq:D1}) to (\ref{eq:D3}) resulting in an optimization problem with bilinear SOS constraints and a non-convex objective function. \He{ The  vector inequality in  (\ref{eq:D3}) represents many scalar inequalities. Denote row $i$ of $A$ by $A_i$ and element $i$ of $b$ as $b_i$.}
\begin{opt}\label{opt1} ($sosopt_1$) Fix $\epsilon > 0$.
	\begingroup
	\allowdisplaybreaks
	\begin{align}
	&\sup_{V,k,s} \ \text{volume}(\Omega_{t_0,\gamma}^V) \nonumber \\
	&s.t. \ s_2(t,x), s_3(t,x), (s_4(t, x) - \epsilon), s_7(t,x) \in \Sigma[t, x], \nonumber \\
	& s_{i,5}(t,x), s_{i,6}(t,x) \in \Sigma[t,x], \forall i = 1, ..., n_p, \nonumber \\
	& k(t,x) \in \mathbb{R}^m[t,x], V(t,x) \in \mathbb{R}[t,x], \tag{E.1} \label{eq:E1} \\
	&-\left(\frac{\partial V}{\partial t} + \frac{\partial V}{\partial x}(f(t,x) + g(t,x)k(t,x)) \right) - s_2(t,x)h(t) +s_3(t,x)(V(t,x) - \gamma ) \in \Sigma[t,x], \tag{E.2} \label{eq:E2} \\
	&-s_4(t, x) r(t, x) + V(t,x) - \gamma - s_7(t,x)h(t) \in \Sigma[t, x], \tag{E.3} \label{eq:E3} \\
	& \He{b_i(t,x) - A_i(t,x)k(t,x) + s_{i,5}(t,x)(V(t,x) - \gamma)}   \He{- s_{i,6}(t,x)h(t) \in \Sigma[t,x], \forall i = 1, ..., n_p,} \tag{E.4} \label{eq:E4}
	\end{align}
	\endgroup
\end{opt}  
\noindent where the positive number $\epsilon$ ensures that $s_4(t,x)$ is uniformly bounded away from 0. However the choice of $\epsilon$ does affect the optimization, with smaller values of $\epsilon$, theoretically less restrictive. \He{Due to numerical issues, the value must be chosen with care. If $\epsilon$ is too small, numerical issues might arise, but large values cause conservatism. Therefore, trial and error in the selection of $\epsilon$ may be necessary.}
 
 If a target set rather than a target tube is considered, then instead of enforcing constraint (\ref{eq:E3}), the corresponding SOS constraint for (\ref{eq:A3}) is imposed 
 \begin{align}
 -s_a(x) r_T(x) + V(T,x) - \gamma \in \Sigma[x], \tag{E.5} \label{eq:E5}
 \end{align}
 where \He{$(s_a(x) - \epsilon) \in \Sigma[x]$.}
 
In the constraints (\ref{eq:E2}) and (\ref{eq:E4}), there are three bilinear pairs involving decision variables \He{$\left(k, \frac{\partial V}{\partial x}\right)$, $(s_3, V)$, $(s_{i,5}, V)$}, rendering these constraints non-convex. To tackle the non-convex optimization problem, we decompose it into two subproblem, iteratively searching between the reachability storage function $V$ and multipliers / control laws $s,k$. In Algorithm \ref{alg:alg1}, $\epsilon$ is still a fixed small positive number, but $\gamma$ becomes a scalar decision variable.
\begin{algorithm} [H]
	\caption{Iterative method}
	\label{alg:alg1}
	\begin{algorithmic}[1]
	\Require{function $V^0$ such that constraints (E.2 - 4) are feasible by proper choice of $s, k, \gamma$.}
	\Ensure{($k$, $\gamma$, $V$) such that with the volume of $\Omega_{t_0,\gamma}^V$ having been enlarged.}
	\For{$j = 1:N_{iter}$}
		\State \He{$\boldsymbol{\gamma}$\textbf{-step}: decision variables $(s, k,\gamma)$.
			
		Maximize $\gamma$ subject to (E.2 - 4)  using $V = V^{j-1}$. This yields ($s_3^j, s_{i,5}^j, k^j$) and optimal reward $\gamma^j$.}
		\State \He{$\boldsymbol{V}\textbf{-step}$: decision variables $(s_1, s_2, s_4, s_{i,6}, s_7, V)$; 
			
		Maximize the feasibility (analytic center described below) subject to (E.2 - 4) as well as $s_1(x) \in \Sigma[x],$ and
		\begin{align}
		& -(V(t_0,x) - \gamma^j) +  s_1(x) (V^{j-1}(t_0,x) - \gamma^j) \in \Sigma[x], \tag{E.6} \label{eq:E6}
		\end{align}
		
		using ($\gamma = \gamma^j, s_3 = s_3^j, s_{i,5} = s_{i,5}^j, k=k^j$). This yields $V^j$.}
		\EndFor
		\end{algorithmic}
\end{algorithm}

 \begin{remark}
	\He{In the examples of Section \ref{ex_section}, the target region is a neighborhood around an equilibrium point, and a linear state-feedback for the linearization about the equilibrium point was used to compute the initial iterate, $V^0$, \cite{Ufuk:09} \cite{Erin:13}. }
\end{remark}

\begin{remark}
	The global optima in the $\gamma$-step can be computed by bisecting $\gamma$. Since only $(s_3, \gamma)$ and $(s_{i,5}, \gamma)$ enter bilinearly, and $\gamma$ is the objective function, the $\gamma$-step is a generalized SOS problem, which is proven in \cite{Pete:10} to be quasiconvex.  
\end{remark}

\begin{remark}
	\He{After the $\gamma$-step, many of the constraints are active. The subsequent $V$-step is formulated to return the decision variables at the analytic center of the feasible set \cite{Boyd:93} \cite{boyd:04}, pushing the newly computed storage function away from the constraints thus enabling further progress on the next $\gamma$ step.}
\end{remark}	

\begin{remark}\label{grow1}
	\He{	(\ref{eq:E6}) enforces $\Omega_{t_0, \gamma^j}^{V^{j-1}} \subseteq \Omega_{t_0, \gamma^j}^{V^{j}}$, which ensures that the BRS inner-approximation computed by the $j$'th $V$-step at least contains the inner-approximation obtained by the $j$'th $\gamma$-step.}
\end{remark}
\He{
\begin{theorem} \label{grow2}
	The BRS inner-approximation from the $(j+1)$'th $\gamma$-step contains the inner-approximation from the $j$'th $V$-step: $ \Omega_{t_0, \gamma^{j}}^{V^{j}} \subseteq  \Omega_{t_0, \gamma^{j+1}}^{V^{j}}$.
\end{theorem} 
\begin{proof}
	The obtained decision variables $(s_2^j, s_4^j, s_{i,6}^j, s_7^j, V^j)$ from the $j$'th $V$-step along with the fixed values (from the $j$'th $\gamma$-step) $(\gamma^j, s_3^j, s_{i,5}^j, k^j)$, are feasible for (E.2 - 4), and thus are feasible for the $(j+1)$'th $\gamma$-step. Since $\gamma^{j+1}$ is the optimal reward of the $(j+1)$'th $\gamma$-step, it gives $\gamma^{j+1} \ge \gamma^j$.
\end{proof}

From Remark \ref{grow1} and Theorem \ref{grow2} we can conclude that quality of the BRS inner-approximation will improve with each iteration.
}
\begin{remark} \label{noglobal}
	\He{Coordinate-wise algorithms do not in general converge to the global optima. Thus} although  the subproblems in the $\gamma$-step and $V$-step at each iteration are solved exactly, the iterative algorithm does not necessarily yield the global optimal solution for the optimization $sosopt_{\ref{opt1}}$.  
\end{remark}

\section{Incorporating System Uncertainties}
Two different sources of uncertainty are addressed. Uncertainties with $\mathcal{L}_2$ bounds, denoted as $w$, are used to model external disturbances. Time-varying uncertainties with $\mathcal{L}_{\infty}$ bounds, denoted as $\delta$, are used to model uncertain parameters in the system. Thus the dynamical system is
\begin{align}
\dot{x}(t) = f(t,x(t),w(t),\delta(t)) + g(t,x(t),w(t),\delta(t))u(t), \label{eq:system2}
\end{align}
with $w(t) \in \mathbb{R}^{n_w}$, $\delta(t) \in \mathbb{R}^{n_{\delta}}$, and \He{polynomial vector field $f \in \mathbb{R}^n[t,x,w,\delta]$,  $g \in \mathbb{R}^{n \times m}[t,x,w,\delta]$}. 

The assumptions on $\delta$ and $w$ are as follows. The parametric uncertainties $\delta(t)$ belong to the set $\Delta_{\overline{\delta}} := \left\{\delta \in \mathbb{R}^{n_{\delta}} \vert \delta^T\delta \leq \overline{\delta}^2 \right\}$. A non-decreasing polynomial function $q$ satisfying $q(t_0)=0$, $q(T)=1$ describes how fast the energy of $w$ can be released. Specifically, disturbances $w$ satisfy $\int_{t_0}^t w(\tau)^T w(\tau) d\tau \le R^2 q(t), \forall t \in [t_0, T]$. The quantities $\bar{\delta}$, $R$ and $q(\cdot)$ are assumed to be given.
\begin{theorem} \label{thm3}
	Given system (\ref{eq:system2}), initial time $t_0$, terminal time $T \ge t_0$, \He{a function $r$ and associated target tube $\Omega_{t,0}^{r}$}, bounds $\bar{\delta}$, $R$ and function $q(\cdot)$. If there exists a $\mathcal{C}^1$ function $V: \mathbb{R} \times \mathbb{R}^n \rightarrow \mathbb{R}$,  and a control law $k: \mathbb{R} \times \mathbb{R}^n \times \mathbb{R}^{n_w} \times \mathbb{R}^{n_{\delta}} \rightarrow \mathbb{R}^m$, such that for all $ (t,w,\delta) \in [t_0, T]\times\mathbb{R}^{n_w}\times \Delta_{\overline{\delta}}$,
		\begin{align}
	&\Omega_{t,\gamma+R^2q(t)}^V \subseteq \bigg\{x \in \mathbb{R}^n \bigg\vert \frac{\partial V}{\partial t} + \frac{\partial V}{\partial x}( f(t,x,w,\delta) +  g(t,x,w,\delta)k(t,x,w,\delta)) \leq w^T w \bigg\} \tag{F.1} \label{eq:F1} 
	\end{align}
	and for all $\ t \in [t_0, T]$,
	\begin{align}
	 \Omega^V_{t,\gamma+R^2q(t)} \subseteq \Omega_{t,0}^{r}, \tag{F.2} \label{eq:F2} 
	\end{align}
	then for all $x(t_0) \in \Omega_{t_0, \gamma}^V$, $x(t) \in \Omega_{t,0}^{r}$, for all $t \in [t_0, T]$, under the control law $k$. 
\end{theorem}

\begin{proof}
	By assumption $x(t_0) \in \Omega_{t_0,\gamma}^V$, then we have $V(t_0, x(t_0)) \leq \gamma$. Integrating the dissipation inequality in (\ref{eq:F1}), we have $V(t,x(t)) \leq V(t_0, x(t_0)) + \int_{t_0}^t w(\tau)^T w(\tau) d\tau \leq \gamma + \int_{t_0}^t w(\tau)^T w(\tau) d\tau \leq \gamma + R^2q(t), \forall t \in [t_0, T]$, and it follows from (\ref{eq:F2}) that $x(t) \in \Omega_{t,0}^r, \forall t \in [t_0, T]$. 
\end{proof}
\He{
\begin{remark}\label{k_restrict}
	In Theorem \ref{thm3}, the control law is allowed to depend on $t,x,w,\delta$. Restricting the dependence of $k$ is straightforward, as the theorem remains true if $k$ depends on a subset of the variables $t,x,w,\delta$. 
\end{remark}
}

 If $q$ is not known beforehand, meaning $w$ only satisfies $\int_{t_0}^T w(\tau)^T w(\tau)d\tau \le R^2$, the constraints (\ref{eq:F1}) and (\ref{eq:F2}) need to be modified:  for all $(t,w,\delta) \in [t_0, T]\times\mathbb{R}^{n_w}\times \Delta_{\overline{\delta}}$, 
\begin{align}
&\Omega_{t,\gamma+R^2}^V \subseteq \bigg\{x \in \mathbb{R}^n \bigg\vert \frac{\partial V}{\partial t} + \frac{\partial V}{\partial x}( f(t,x,w,\delta) +  g(t,x,w,\delta)k(t,x,w,\delta)) \leq w^T w \bigg\} \nonumber 
\end{align}
and for all $t \in [t_0, T]$, $\Omega^V_{t,\gamma+R^2} \subseteq \Omega_{t,0}^{r}$.  

Control saturation is again addressed by adding appropriate constraints. \He{This culminates in a state-feedback synthesis BRS optimization that accounts for actuator limits, external disturbances, and parametric uncertainties,}
\begin{highopt}\label{highopt3} ($hi$-$opt_3$)
	\begingroup
	\allowdisplaybreaks
	\begin{align}
	&\sup_{V,k} \ \text{volume}(\Omega_{t_0,\gamma}^V) \nonumber \\
	& s.t.\ \Omega_{t, \gamma+R^2q(t)}^V \subseteq \bigg\{x\in \mathbb{R}^n \bigg \vert\frac{\partial V}{\partial t} + \frac{\partial V}{\partial x} (f(t,x,w,\delta) +  g(t,x,w,\delta)k(t,x,w,\delta)) \leq w^Tw\bigg\}, \nonumber \\
	& \quad \quad \quad \forall \ (t,w,\delta) \in [t_0, T]\times \mathbb{R}^{n_w} \times \Delta_{\overline{\delta}},  \tag{G.1} \label{eq:G1} \\
	&\Omega_{t,\gamma + R^2q(t)}^V \subseteq \Omega_{t,0}^{r}, \ \forall \ t \in [t_0, T], \tag{G.2} \label{eq:G2}\\
	&\He{\Omega_{t,\gamma+R^2q(t)}^V \subseteq \left\{x \in \mathbb{R}^n \middle\vert A_i(t,x)k(t,x,w,\delta) \leq b_i(t,x) \right\},} \ \forall\ (t,w,\delta) \in [t_0, T]\times\mathbb{R}^{n_w} \times \Delta_{\overline{\delta}}, \forall i = 1, ..., n_p. \tag{G.3} \label{eq:G3}
	\end{align}
	\endgroup
\end{highopt}

If, in addition to the $\mathcal{L}_2$ bound, $w$ satisfies an $\mathcal{L}_{\infty}$ constraint, $w(t) \in \Delta_{\overline{w}} := \{ w \in \mathbb{R}^{n_w} \vert w^Tw \leq \overline{w}^2\}$, then constraints (\ref{eq:G1}) (\ref{eq:G3})  are only restricted to hold for all $(t, w, \delta) \in [t_0, T] \times \Delta_{\overline{w}} \times \Delta_{\overline{\delta}}$. 

Applying SOS relaxation and the S-procedure to $hi$-$opt_{\ref{highopt3}}$, yields the following optimization problem. Again, $\epsilon$ is a fixed small positive number.
\begin{opt}($sosopt_2$) \label{opt2} 
	\begingroup
	\allowdisplaybreaks
	\begin{align}
	&\sup_{V,k,s} \ \text{volume}(\Omega_{t_0,\gamma}^V) \nonumber \\
	&s.t. \ s_{l}(t,x,w,\delta) \in \Sigma[t,x,w,\delta], \forall l = 2,3,8,9, \nonumber \\
	&(s_4(t,x) - \epsilon), s_7(t,x) \in \Sigma[t,x], \nonumber \\
	& s_{i,j}(t,x,w,\delta) \in \Sigma[t,x,w,\delta], \forall i = 1, ..., n_p, \forall j = 5,6,10,11,\nonumber \\
	&k(t,x,w,\delta) \in \mathbb{R}^m[t,x,w,\delta], V(t,x) \in \mathbb{R}[t,x], \tag{H.1} \label{eq:H1} \\
	&-\bigg(\frac{\partial V}{\partial t} + \frac{\partial V}{\partial x}(f(t,x,w,\delta) + g(t,x,w,\delta) k(t,x,w,\delta))  - w^Tw \bigg) +s_3(t,x,w,\delta)(V(t,x) - \gamma -R^2q(t)) \nonumber \\
	&\quad \quad - s_2(t,x,w,\delta)h(t) + s_{8}(t,x,w,\delta)(w^Tw - \overline{w}^2) +s_{9}(t,x,w,\delta)(\delta^T\delta - \overline{\delta}^2 ) \in \Sigma[t,x,w,\delta], \tag{H.2} \label{eq:H2} \\
	&-s_4(t,x) r(t,x) + V(t,x) - \gamma - R^2q(t) - s_7(t,x)h(t)\in \Sigma[t,x], \tag{H.3} \label{eq:H3} \\
	& \He{b_i(t,x) - A_i(t,x)k(t,x,w,\delta) + s_{i,11}(t,x,w,\delta)(w^T w - \overline{w}^2 )}  \He{+ s_{i,5}(t,x,w,\delta)(V(t,x) - \gamma - R^2q(t)) } \nonumber \\
	&\quad \quad  \He{- s_{i,6}(t,x,w,\delta)h(t) + s_{i,10}(t,x,w,\delta)(\delta^T\delta - \overline{\delta}^2 )}  \He{ \in \Sigma[t,x,w,\delta], \forall i = 1, ..., n_p.}   \tag{H.4} \label{eq:H4}
	\end{align}
	\endgroup
\end{opt}
By slightly modifying Algorithm \ref{alg:alg1}, an iterative algorithm for $sosopt_{\ref{opt2}}$ is developed.

\begin{remark}
	 As mentioned in Remark \ref{k_restrict}, the dependence of $k$ can be more restrictive and the multipliers simplify. For example, with $k(t,x)$, $s_{i,10}$ and $s_{i,11}$ can be eliminated and $s_{i,5}$ and $s_{i,6}$ only need to depend on $t$ and $x$. \He{Example \ref{ex_GTM} illustrates this flexibility.}
\end{remark}

\section{Examples}\label{ex_section}
A workstation with a 2.7 [GHz] Intel Core i5 64 bit processor and 8[GB] of RAM was
used for performing all computations in the following examples. The SOS optimization
problem is formulated and translated into SDP using the sum-of-square module in SOSOPT \cite{Pete:13} on MATLAB, and solved by the SDP solver MOSEK \cite{Mosek:17}. Table \ref{tab:table} shows the degree of various
polynomials and the computation time.
\begin{table}[h]
	\caption{Computation times for each example \label{tab:table}}
	\centering
	\begin{tabular}{ |M{3.6cm}|M{2.0cm}|M{2.0cm}|M{1.4cm}|M{1.6cm}|M{3.8cm}|}
		\hline
		Examples / sections & Number of States &Degree of Dynamics &Degree of $V$ &Degree of $s, k$ &Computing Time [sec] \\
		\hline
		\ref{ex_car}  &3 &1 & 6 &2  & $7.2 \times 10^3$   \\ \hline
	     \ref{ex_car_obs}  &3 &1 & 6 &2  & $8.7 \times 10^3$   \\ \hline
		 \ref{ex_pendubot}   &4&3&4&4&$1.3 \times 10^4$ \\ \hline
		 \ref{ex_GTM}: GTM without $w$  &4&3&4&4& $1.1 \times 10^4$ \\ \hline
		 \ref{ex_GTM}: GTM with $w$, $k(t,x)$   &4&3&4&4& $2.1 \times 10^4$ \\ \hline
		 \ref{ex_GTM}: GTM with $w$, $k(t,x,w)$   &4&3&4&4& $4.8 \times 10^4$ \\ \hline
		 \ref{ex_game}  &3&3&6&2& $7.2 \times 10^3$ \\ \hline
	\end{tabular}
\end{table}
%

\subsection{Dubin's Car} \label{ex_car}
Consider the Dubin's car \cite{Dubins:57}, a multi-input system: $\dot{a} = v \cos(\theta), \dot{b} = v \sin(\theta), \dot{\theta} = \omega,$ with states $a$: $x$ position (m), $b$: $y$ position (m), $\theta$: yaw angle (rad) and control inputs $\omega$: turning rate (rad/s), $v$: forward speed (m/s). By the change of variables, $x_1 = \theta$, $x_2 = a \cos(\theta)+b \sin(\theta)$, $x_3 = -2(a \sin(\theta) - b \cos(\theta))+\theta x_2$, and $u_1 = \omega$, $u_2 = v - \omega(a \sin(\theta) - b \cos(\theta))$, it is transformed into polynomial dynamics \cite{David:07}: 

\begin{equation}
\begin{aligned}
&\dot{x}_1 = u_1, \\ 
&\dot{x}_2 = u_2, \\ 
&\dot{x}_3 = x_2 u_1 - x_1 u_2. 
\end{aligned}
\end{equation}
We take $[t_0, T] = [0, 4 \ \text{sec}]$, $r_T(x) = x^T x - 0.2^2$, $\epsilon = 1\times10^{-3}$, and impose bounds on control inputs $u_1, u_2 \in [-1, 1]$. A closed-loop simulation with the resulting controller and initial condition $[-0.8, 1.4, 0.3]$ is shown in Figure \ref{fig:carSim}.  Figure \ref{fig:car_obs} shows the slices of sets with $x_3 = 0$, $x_2 = 0$, and $x_1 = 0$, respectively. $\Omega_{t_0,\gamma}^V$ is shown as the dashed curves and $\Omega_0^{r_T}$ is shown as the dash-dot curves. 
 
\begin{figure}[h]
	\centering
	\includegraphics[width=0.5\textwidth]{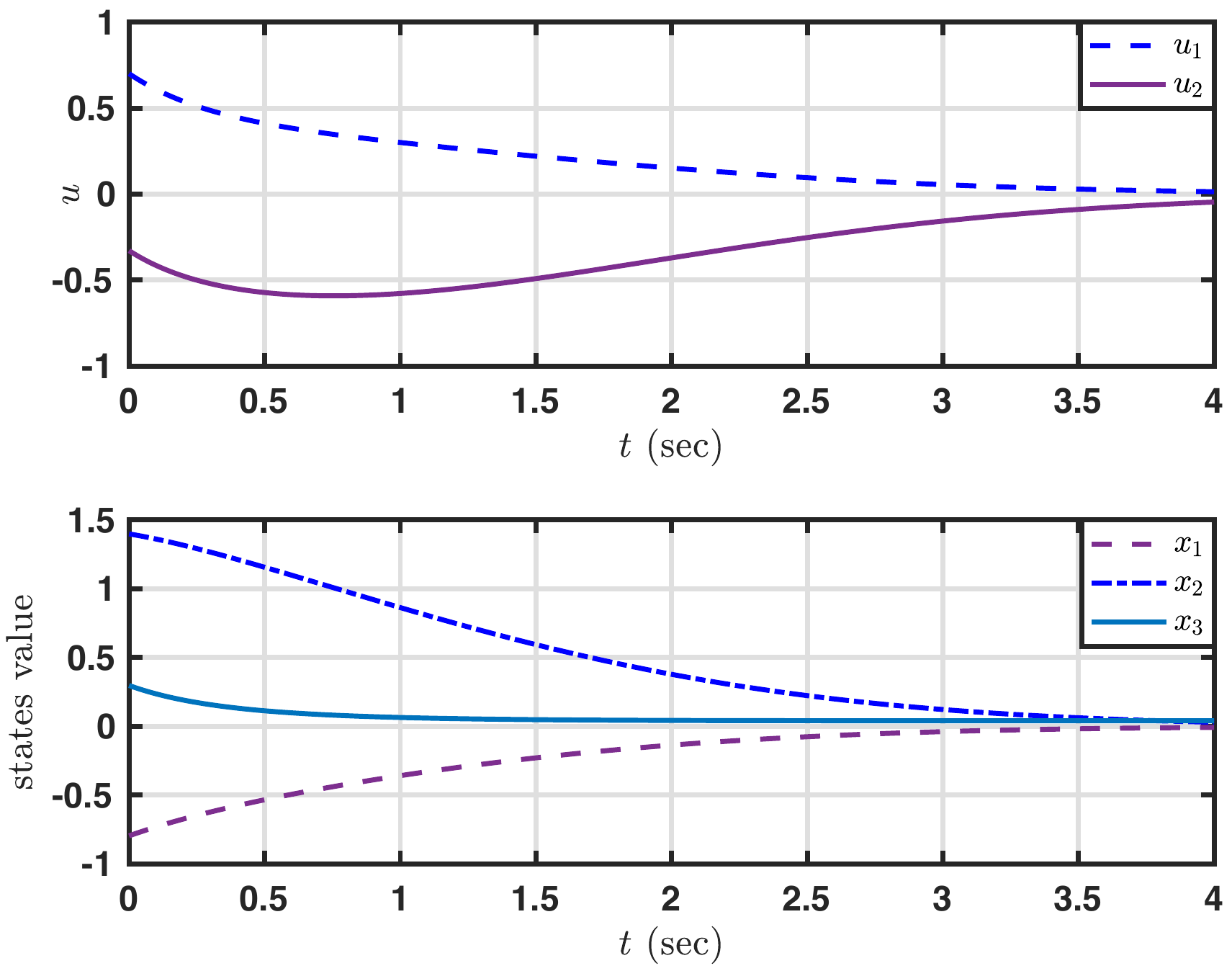}
	\caption{Dubin's car Simulations}
	\label{fig:carSim}    
\end{figure}

\subsubsection{Dubin's Car with Obstacle} \label{ex_car_obs}
In addition to the terminal target set $\Omega_{0}^{r_T}$, suppose there is an unsafe region $\Omega_0^{obs} := \{x \in \mathbb{R}^3 | obs(x) := (x_1 - 1.5)^2 + x_2^2 + x_3^2 - 0.5^2 \leq 0\}$. Thus the target tube is the intersection of the terminal target set $\Omega_0^{r_T}$ and the complement of $\Omega_{0}^{obs}$. Slices of the resulting $\Omega_{t_0,\gamma}^V$ with the obstacle are shown as solid black curves in Figure \ref{fig:car_obs}.
\begin{figure}[h]
	\centering
	\includegraphics[width=0.7\textwidth]{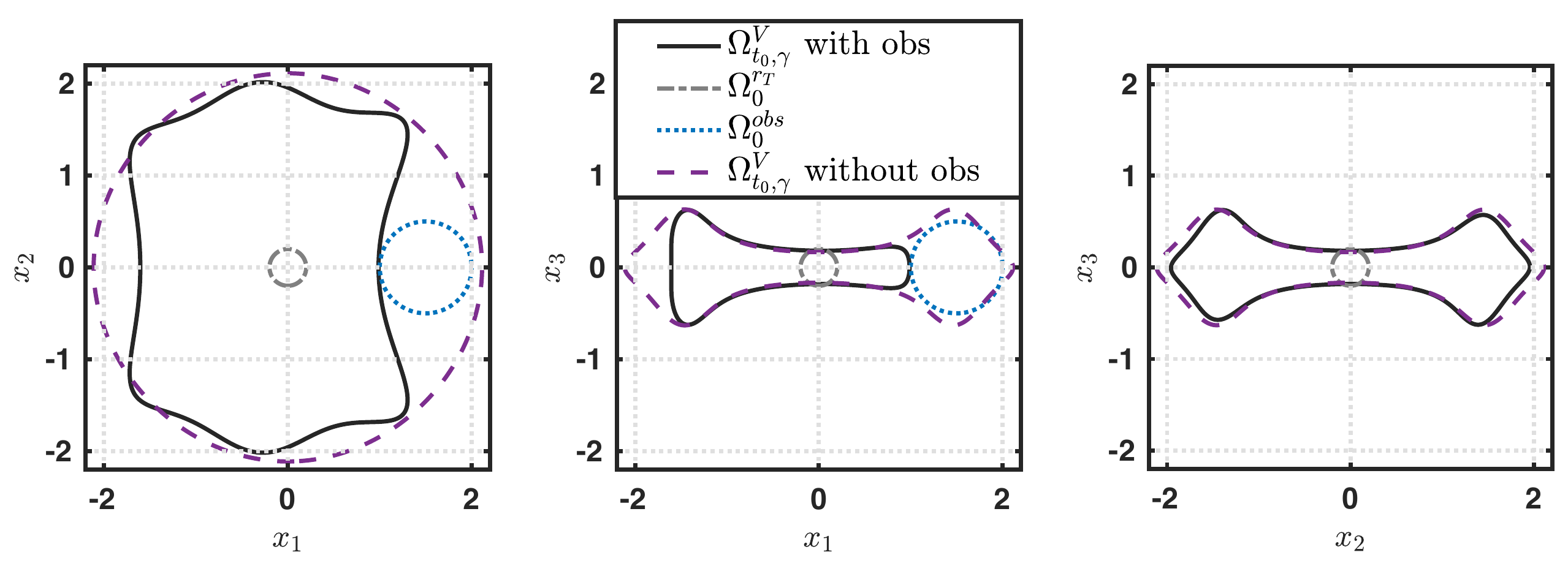}
	\caption{Inner-approximated BRS for Dubin's car example}
	\label{fig:car_obs}    
\end{figure}

\subsection{Pendubot Example}\label{ex_pendubot}
Consider the following polynomial dynamics for a pendubot
\begingroup
\allowdisplaybreaks
\begin{align}
\bmat{\dot{x}_1 \\ \dot{x}_2 \\ \dot{x}_3 \\ \dot{x}_4} = \bmat{x_2 \\ f_2(x_1,x_2,x_3,x_4) \\ x_4 \\ f_4(x_1,x_2,x_3,x_4)} + \bmat{0 \\ g_2(x_3) \\ 0 \\ g_4(x_3)}u, \nonumber
\end{align}
\endgroup
with
\begingroup
\allowdisplaybreaks
\begin{align}
f_2 &= - 10.656 x_1^3 + 11.531 x_1^2 x_3 + 7.885 x_1 x_3^2 + 0.797 x_2^2 x_3  + 0.841 x_2 x_3 x_4 + 21.049 x_3^3  + \nonumber \\
&\quad \quad \quad 0.420 x_3 x_4^2 + 66.523 x_1 - 24.511 x_3, \nonumber \\
f_4 &= 10.996 x_1^3 - 48.915 x_1^2 x_3 - 6.404 x_1 x_3^2 - 2.396 x_2^2 x_3  - 1.594 x_2 x_3 x_4 - 51.909 x_3^3 - \nonumber \\
&\quad \quad \quad 0.797 x_3 x_4^2 - 68.642 x_1 + 103.978 x_3, \nonumber \\
g_2 &= -10.096 x_3^2 + 44.252, \nonumber \\
g_4 &=  37.802 x_3^2 - 83.912, \nonumber
\end{align}
\endgroup
which is obtained as a least-squares approximation of the full equations for $x_1 \times x_3 \in [-1, 1] \times [-1, 1]$.

Here $x_1$ and $x_3$ represent $\theta_1$ (rad) and $\theta_2$ (rad), which are angular positions of the first link and the second link (relative
to the first link), respectively, and $x_2$ and $x_4$ are $\dot{\theta}_1$ (rad/s) and $\dot{\theta}_2$ (rad/s), which are corresponding angular velocities. Input $u$ (Nm) is the torque applied at the joint of first link and ground, but there is no torque applied at the joint of two links. 

We take the time horizon $[0, 4 \ \text{sec}]$, $r_T(x) = x^T$ $diag(1/0.1^2,$ $1/0.35^2,$ $1/0.1^2,$ $1/0.35^2) x - 1$, $\epsilon = 1\times 10^{-4}$, and impose a bound on the control input $u \in [-1, 1]$. Slices of sets shown on the left side of Figure \ref{fig:pendubot} are plotted with $\dot{\theta}_1$ and $\dot{\theta}_2$ fixed at $0$. Slices of sets shown on the right side of Figure \ref{fig:pendubot} are plotted with $\theta_1$ and $\theta_2$ fixed at $0$. 
\begin{figure}[h]
	\centering
	\includegraphics[width=0.6\textwidth]{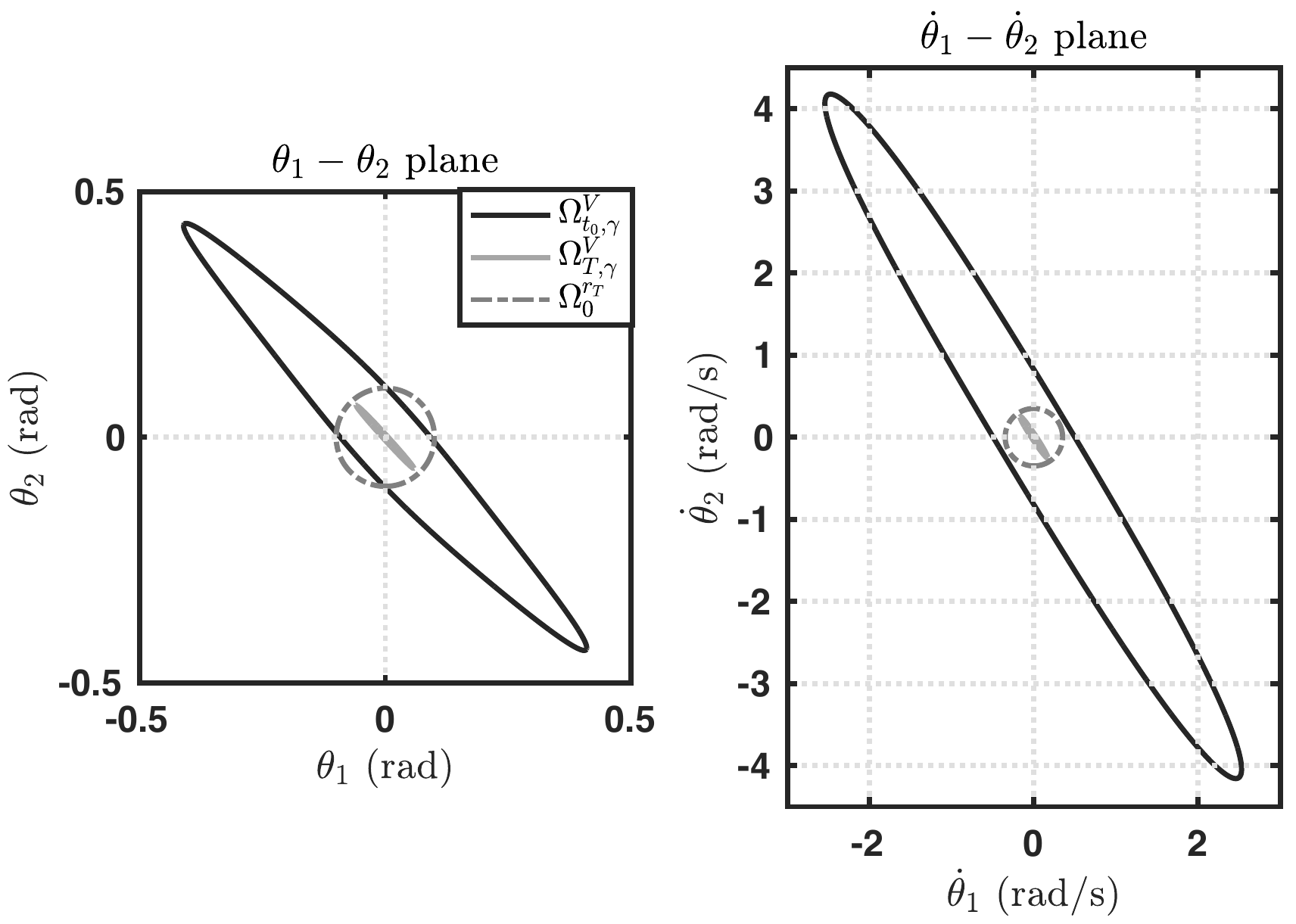}
	\caption{Inner-approximated BRS for the pendubot example}
	\label{fig:pendubot}    
\end{figure}

A simulation result with the initial condition [$-0.35$ rad; $2.6$ rad/s; $0.35$ rad; $-4$ rad/s], under the designed polynomial control law is shown in Figure \ref{fig:pendubotSim}.
\begin{figure}[h]
	\centering
	\includegraphics[width=0.48\textwidth]{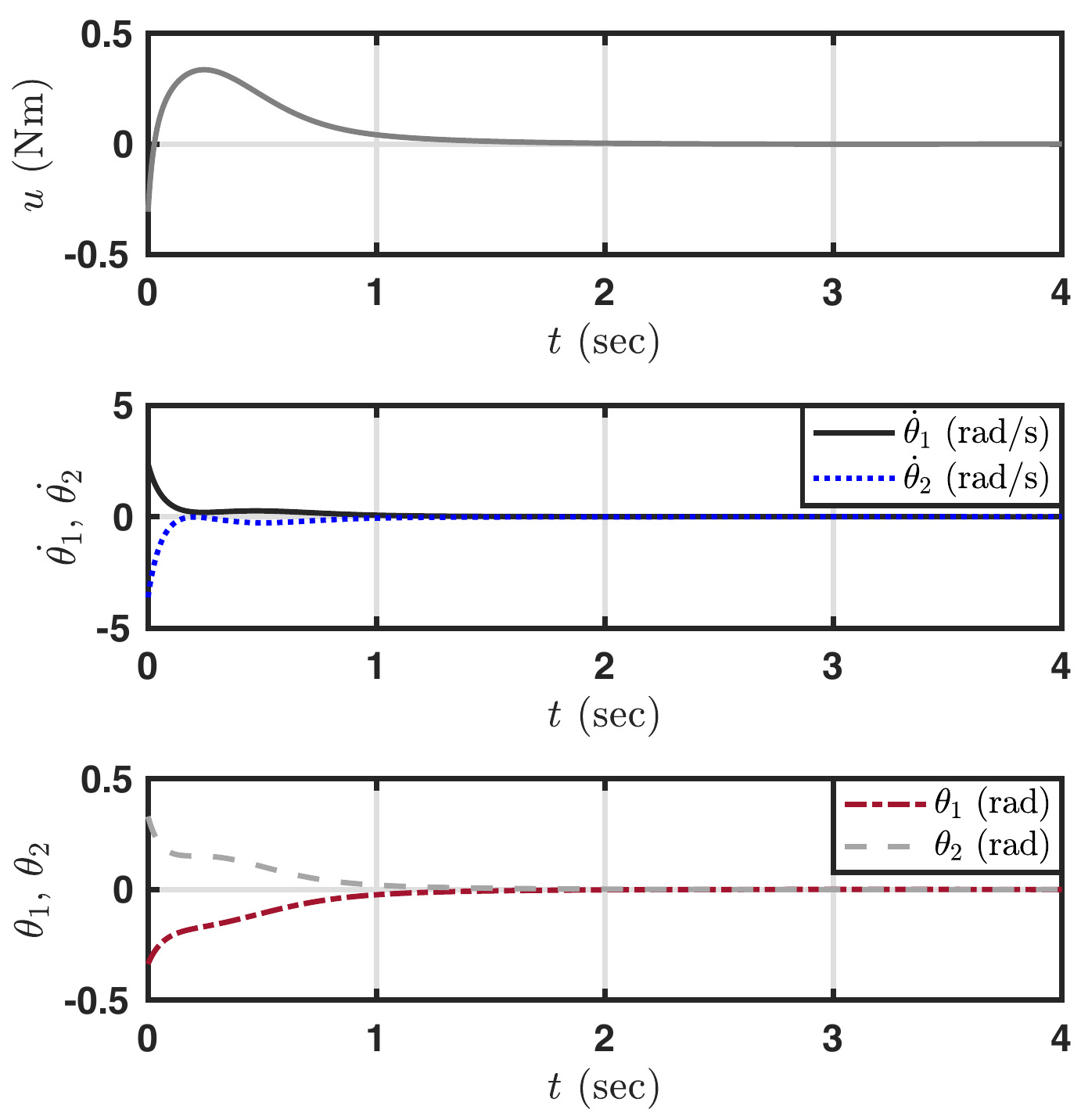}
	\caption{Pendubot simulation results}
	\label{fig:pendubotSim}    
\end{figure}

\subsection{NASA's Generic Transport Model (GTM) around straight and level flight condition with $\mathcal{L}_2$ Disturbance}\label{ex_GTM}
The GTM is a remote-controlled 5.5\% scale commercial aircraft \cite{Murch:07}. From \cite{Brain:92}, its longitudinal dynamical model is 
\begin{equation}
\begin{aligned}
\dot{x}_1 = & \frac{1}{m}(-D -mg\sin(x_4 - x_2) + T_x \cos(x_2)  +  T_z \sin(x_2)),   \\
\dot{x}_2 = & \frac{1}{m x_1}(-L +mg \cos(x_4 - x_2)  - T_x \sin(x_2) + T_z \cos(x_2) + x_3), \label{eq:systemGTM}\\
\dot{x}_3 = & \frac{M + T_m}{I_{yy}},  \\
\dot{x_4} = & x_3,  
\end{aligned}
\end{equation}
where $x_1$ to $x_4$ represent air speed (m/s), angle of attack (rad), pitch rate (rad/s) and pitch angle (rad), respectively. The control inputs are elevator deflection $u_{elev}$ (rad) and engine throttle $u_{th}$ (percent). The drag force $D$ (N), lift force $L$ (N), and aerodynamic pitching moment $M$ (N m) are given by $D = \bar{q}SC_D(x_2, u_{elev}, \hat{q})$, $L=\bar{q}SC_L(x_2,u_{elev},\hat{q})$, and $M = \bar{q}S\bar{c}C_m(x_2, u_{elev}, \hat{q})$,where $\bar{q}:= \frac{1}{2}\rho x_1^2$ is the dynamic pressure (N/m$^2$), $\hat{q}:=(\bar{c}/2x_1)x_3$ is the normalized pitch rate (unitless), $S$ and $\bar c$ are the surface area and mean aerodynamic chord (both in m). $C_D, C_L,$ and $C_m$ are aerodynamic coefficients computed from look-up tables provided by NASA \cite{Chakraborty:2011}.

A 4-state, 2-input, degree-7 polynomial model is obtained in \cite{Chakraborty:2011} by replacing all nonpolynomial terms in (\ref{eq:systemGTM}) with their polynomial approximations. The following straight and level trim-condition is computed for this model: $x_{1,t} = 45$ m/s, $x_{2,t} = 0.04924$ rad, $x_{3,t} = 0$ rad/s, $x_{4,t} = 0.04924$ rad, with $u_{elev,t} = 0.04892$ rad, and $u_{th,t} = 14.33 \%$. 
A 4-state, degree-3, single-input polynomial longitudinal model is extracted from the 4-state, 2-input, degree-7 polynomial model by holding $u_{th}$ at its trim value, and retaining terms up to degree-3. This degree-3 polynomial model is used for the following synthesis.

The disturbance $w$ is the perturbation to the angle of attack caused by a change in wind direction, i.e. the force generated on the aircraft is due to wind coming at an angle $(x_2 +w)$. Denote the nominal GTM system as $F(x,u) := f(x) + g(x)u$; then the disturbed system is given as
\begingroup
\allowdisplaybreaks
\begin{align}
\dot{x} & = F(x,u) + \frac{dF(x,u)}{d x_2}w \nonumber \\
& = f(x) + \frac{f(x)}{d x_2}w + (g(x)+\frac{dg(x)}{d x_2}w)u.
\end{align}
\endgroup
The disturbance $w$ is assumed to have both $\mathcal{L}_2$ and $\mathcal{L}_{\infty}$ bounds: $R := 0.1$ rad, $\int_{0}^t w^T(\tau) w(\tau) d\tau \leq R^2 q(t) := R^2t^2/T^2$, for all $t \in [0, 3 \ \text{sec}]$ and $\norm{w(t)}_2 \leq \overline{w} := 0.141$ rad. Set the time horizon [$0, 3$ sec], $\epsilon = 1\times 10^{-4}$, the control constraint $u_{elev} \in [-10^{\circ}, 10^{\circ}]$, and $r_T(x) = (x-x_{eq})^T diag(1/4^2, 1/(\pi/30)^2, 1/(\pi/15)^2, 1/(\pi/30)^2) (x - x_{eq})- 1$, where the equilibrium point $x_{eq} := [x_{1,t}, x_{2,t}, x_{3,t}, x_{4,t}]^T$. 

In this section, we inner-approximate the BRS for three cases: without disturbance $w$ and $k$ is a function of $t,x$; with disturbance $w$ and $k$ is allowed to be a function of $t,x,w$; with disturbance $w$ but $k$ is a function of only $t,x$. Curves shown on the left side of Figure \ref{fig:GTMdisturb} are slices of sets with $x_1 = x_{1,t}$ and $x_4 = x_{4,t}$;  curves shown on the right side are slices of sets with $x_2 = x_{2,t}$ and $x_3 = x_{3,t}$. Notice that the volume of inner-approximations of BRS for the two cases with disturbance are smaller than without disturbance. Moreover, for the two cases with disturbance, the volume of inner-approximations for the case using $k(t,x)$ is smaller than the case using $k(t,x,w)$. 

\begin{figure}[h]
	\centering
	\includegraphics[width=0.65\textwidth]{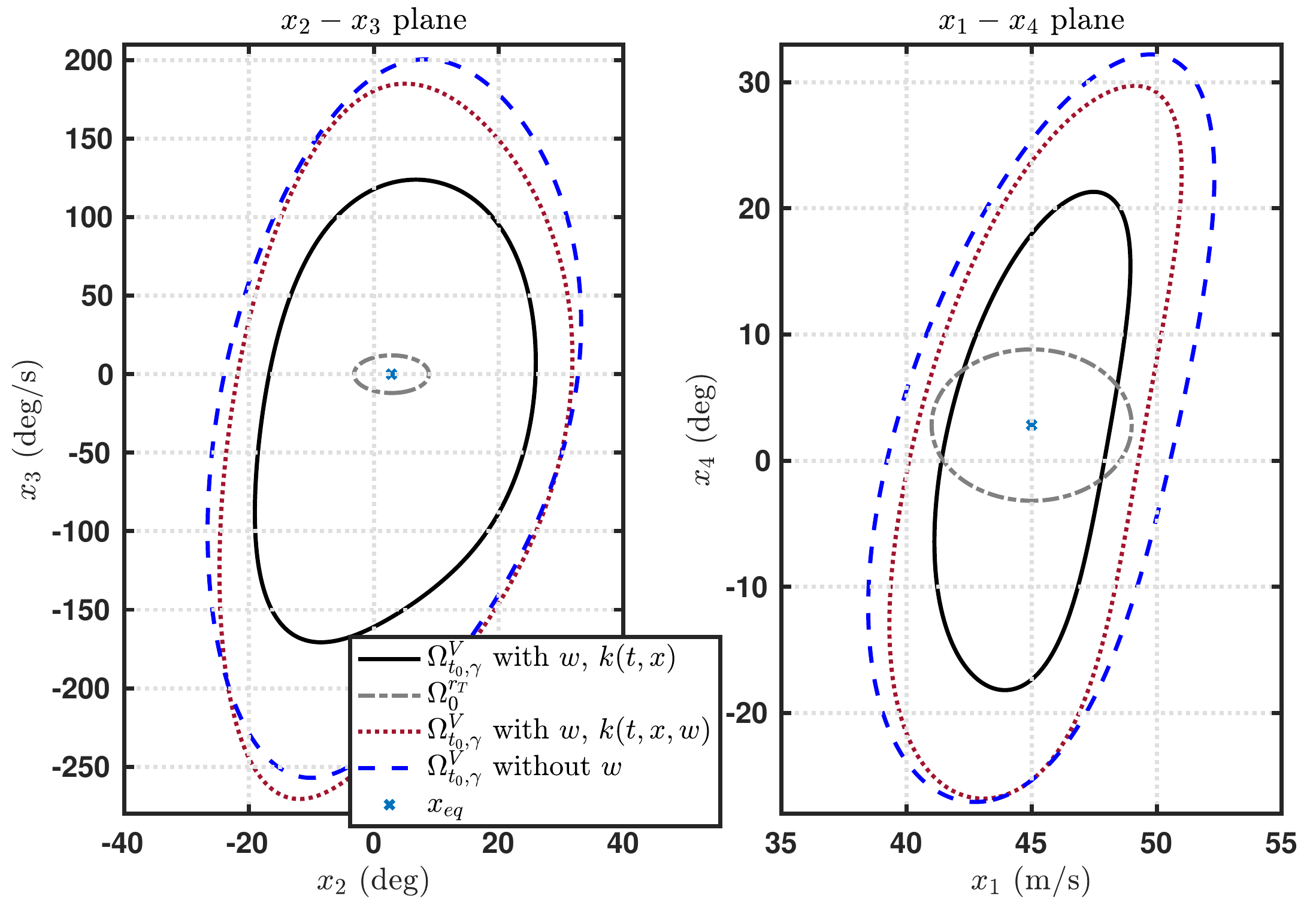}
	\caption{Inner-approximated BRS for GTM}
	\label{fig:GTMdisturb}    
\end{figure}

The simulation results of the polynomial model of GTM with the initial condition [$47$ m/s; $20$ rad; $70$ rad/s; $20$ rad] and a disturbance signal $w(t) = \frac{\sqrt{2t}R}{T} \eta(t)$, using both $k(t,x)$ and $k(t,x,w)$ are shown in Figure \ref{fig:GTMdisturbSim}, where the value of $\eta(t)$ is updated by the number drawn from the uniform distribution on the interval $(-1, 1)$ at 50 Hz, and holds at the updated value until the next update. As we can see in the figure, the trajectory for pitch rate $x_3$ with $k(t,x,w)$ reaches trim value faster than the one with $k(t,x)$, and the former is much smoother.

\begin{figure}[h]
	\centering
	\includegraphics[width=0.55\textwidth]{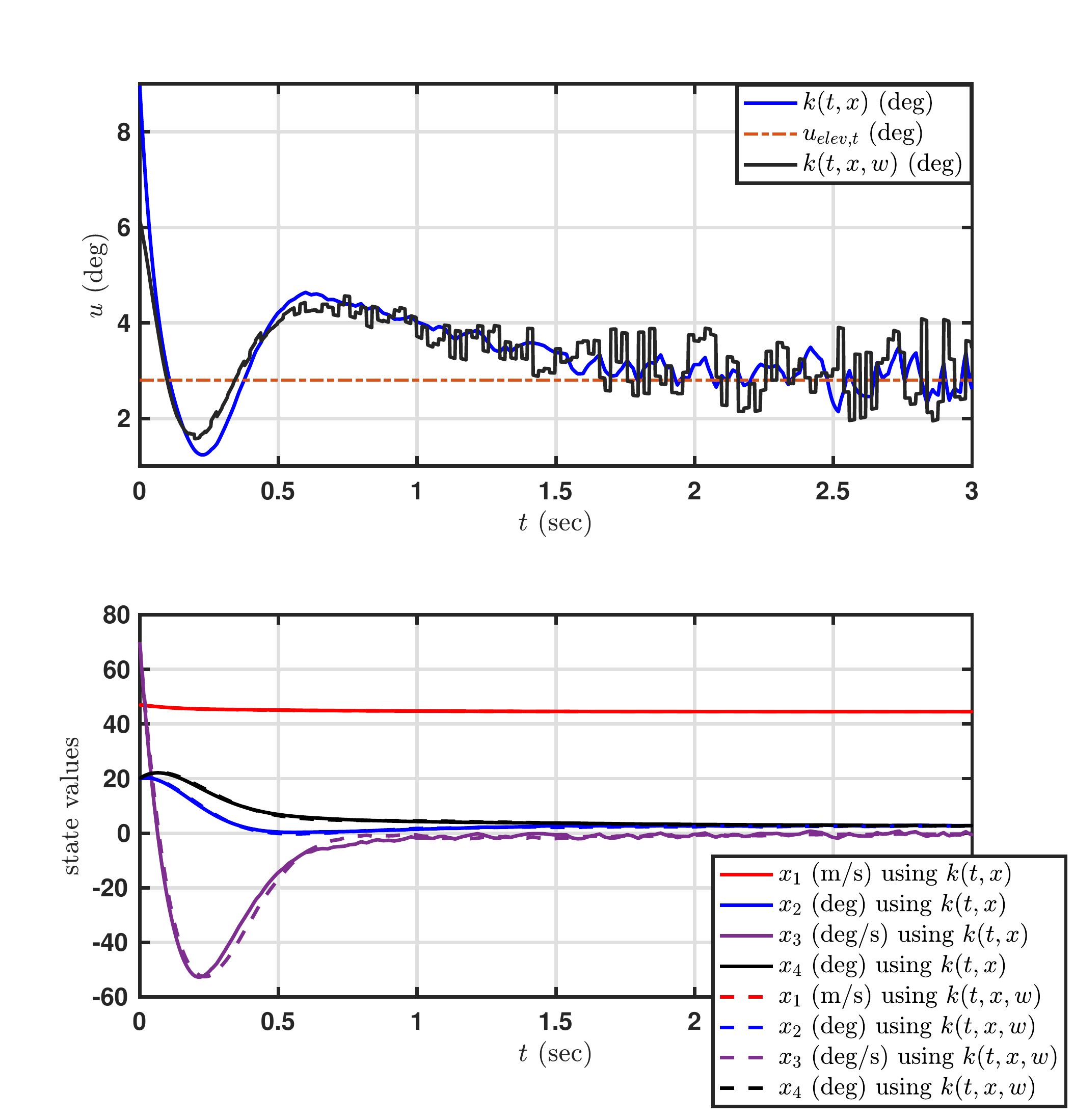}
	\caption{Simulations of GTM with disturbances $w$}
	\label{fig:GTMdisturbSim}    
\end{figure}

\subsection{Pursuer-evader Game} \label{ex_game}
Consider the reach-avoid example from \cite{Ian:05}. Assume that there are two players, the evader and the pursuer. Fix the evader at the origin and facing along the positive $x_1$ axis, so that the pursuer's relative location and heading are described by
\begin{align}
\bmat{\dot{x}_1 \\ \dot{x}_2 \\ \dot{x}_3} = \bmat{-v_e + v_p \cos(x_3) + u_e x_2 \\ v_p \sin(x_3) - u_e x_1 \\ u_p - u_e}, \label{eq:system7} 
\end{align}
where $x_1, x_2, x_3$ represent relative $x, y$ positions and heading angle; $u_e$ and $u_p$ are angular velocity inputs from the evader and pursuer; $v_e$ and $v_p$ are velocities of the evader and pursuer. 

Set the time horizon $[0, 2.6 \ \text{sec}]$, and $r_T(x) = x^T x - 1$. Velocities of two players are constant: $v_e = v_p = 1$, control input is $u_p(t) \in [-1 , 1]$. The goal for the pursuer is to find a robust control law for $u_p$ and an  inner-approximated BRS, so that no matter how the evader chooses its control input at each time instance, all the trajectories for system (\ref{eq:system7}) from the inner-approximated BRS will always be driven to the target set $\Omega_{0}^{r_T}$. This reachability problem is posed as a dynamic game in \cite{Ian:05}, whereas in this paper, the control input $u_e$ from the evader is regarded as the uncertain parameter with a given $\mathcal{L}_{\infty}$ bound: $u_e(t) \in [ -0.5, 0.5]$. In this example,  \He{$\cos(x_3)$ is approximated by $(-0.4298x_3^2 + 1)$, and $\sin(x_3)$ is approximated by $(-0.1511x_3^3 + x_3)$, which are obtained by least square regression  for  $x_3 \in [-\frac{\pi}{2}, \frac{\pi}{2}]$. Polynomial dynamics of (\ref{eq:system7}) can be obtained by replacing $\cos(x_3)$ by $(-0.4298x_3^2 + 1 + \delta_{\cos})$, where accounting for the error between $\cos(x_3)$ and its polynomial approximation yields $\delta_{\cos}(t) \in [-0.05, 0.05]$ for $x_3 \in [-\frac{\pi}{2}, \frac{\pi}{2}]$. The error between $\sin(x_3)$ and its polynomial approximation is very small and it is neglected. Setting $\delta_{\cos}(t) = 0$, neglects the $\cos(x_3)$ error as well.}

\He{The results are computed for the two cases: $\delta_{\cos}(t) \in [-0.05, 0.05]$ or $\delta_{\cos}(t) = 0$. In Figure \ref{fig:ReachAvoid}, computed inner-approximations are shown with solid red and translucent brown, respectively. The computed storage function of the former case is used as the initial iterate $V^0$ for the latter. The target set is shown with the transparent black cylinder. We can see that when $\delta_{\cos}(t) \in [-0.05, 0.05]$, the BRS inner-approximation is smaller than when $\delta_{\cos}(t) = 0$, but robust against the error resulting from polynomial modelling.}

\begin{figure}[h]
	\centering
	\begin{subfigure}[b]{0.2\textwidth}
		\centering
		\includegraphics[width=1.85\textwidth]{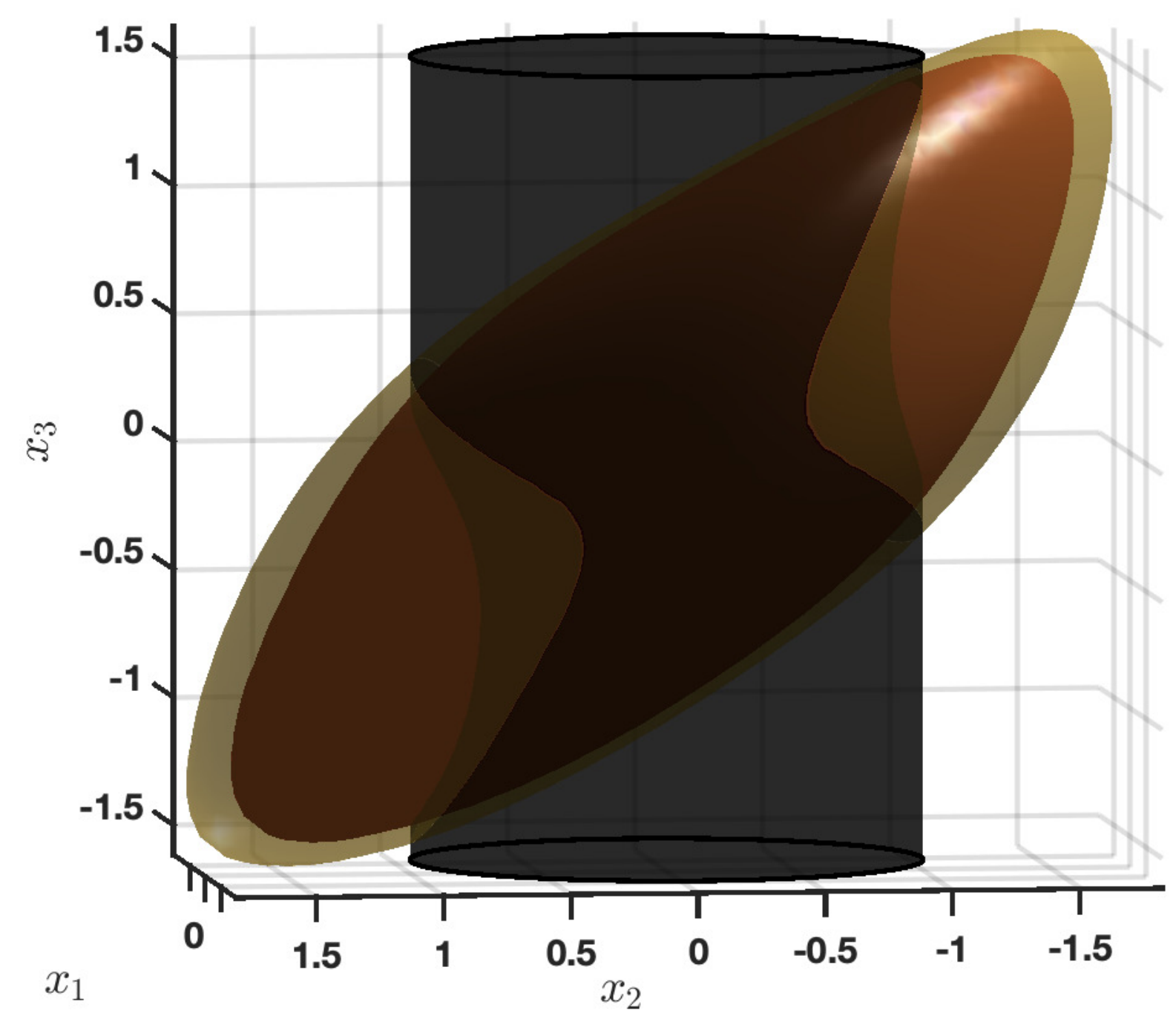}
		\caption{3D view}  
	\end{subfigure}
	\hspace{0.35\textwidth}%
	\begin{subfigure}[b]{0.2\textwidth}  
	    \centering
		\includegraphics[width=1.1\textwidth]{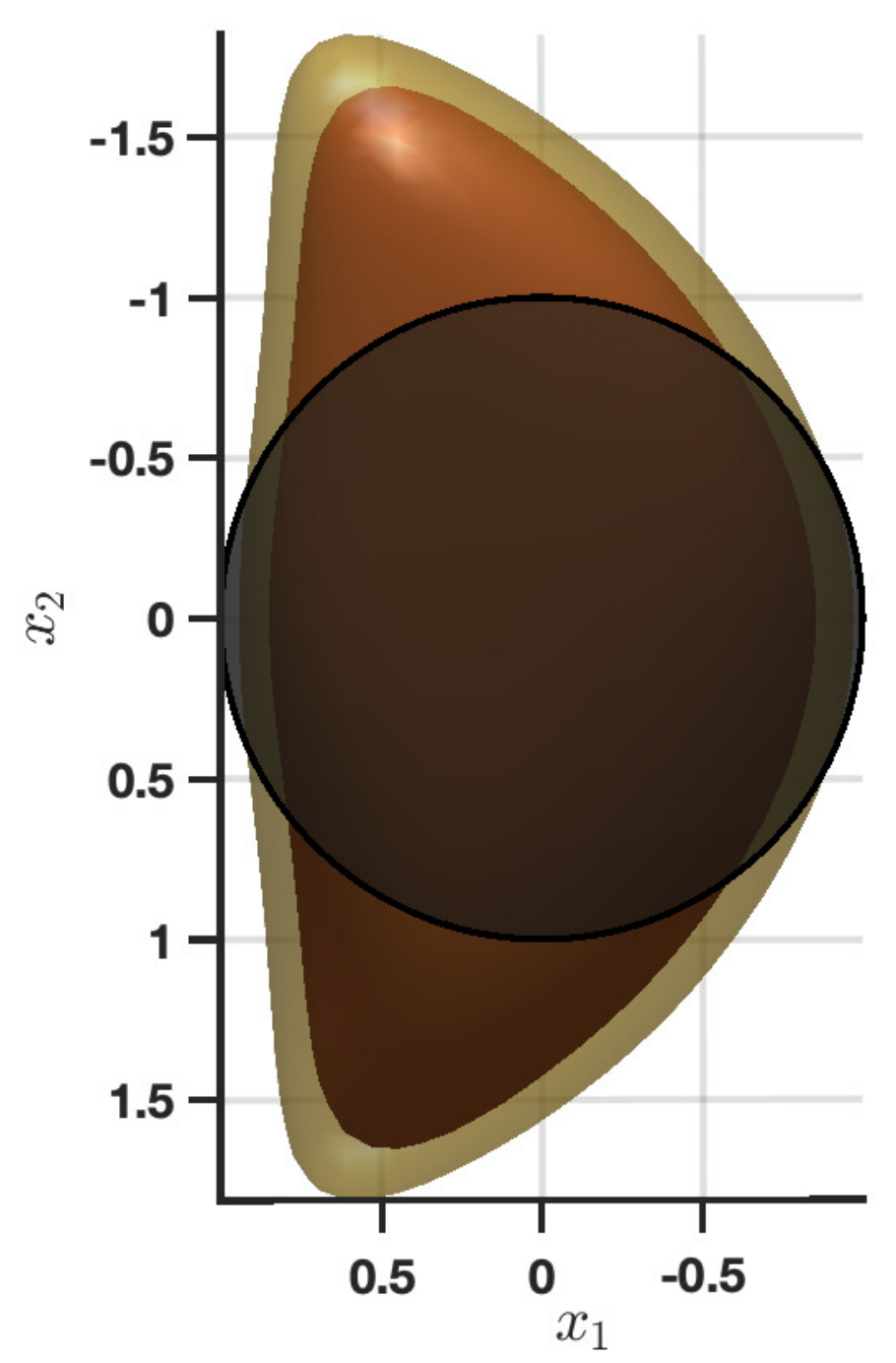}
		\caption{Top view}
	\end{subfigure}
	\caption{Inner-approximated BRS for the pursuer-evader game}
	\label{fig:ReachAvoid}    
\end{figure}

\section{Conclusions}
We proposed a method for synthesizing controllers for nonlinear systems with polynomial vector fields. The synthesis process yields \He{a state-feedback control law, and} a reachability storage function that characterizes an inner-approximation to the BRS \He{for a given target tube}. An iterative algorithm to construct them is derived based on SOS programming and the S-procedure. The synthesis framework is also extended to uncertain systems with $\mathcal{L}_{\infty}$ parametric uncertainties and $\mathcal{L}_2$ disturbances. This method is applied to several practical robotics and aircraft models. \He{Currently, the computational complexity of our method limits it to systems of modest size, with fewer than ten state variables.}

\section*{Acknowledgements}
This work was funded in part by the ONR grant N00014-18-1-2209.

\bibliography{reference.bib} 
\bibliographystyle{IEEEtran}

\end{document}